\documentclass[conference]{IEEEtran}

\usepackage{amssymb}

\usepackage{cite}

\usepackage[pdftex]{graphicx}

\usepackage[cmex10]{amsmath}

\usepackage{array}
\usepackage{mdwmath}
\usepackage{mdwtab}
\usepackage{mdwlist}

\usepackage[font=footnotesize]{subfig}

\usepackage{url}

\usepackage{fixltx2e}
\usepackage{adjustbox}

\usepackage[utf8]{inputenc}

\usepackage{tikz}
\usetikzlibrary{intersections}
\usetikzlibrary{patterns}
\usetikzlibrary{shapes,backgrounds}

\usepackage{booktabs}

\usepackage{algorithm}
\usepackage[noend]{algpseudocode}

\begin{document}
	
\newcommand{\minpoints}{mpts}
\newcommand{\mrgprime}[1][\minpoints]{{G'}_{#1}}
\newcommand{\rng}[1]{RNG^{#1}}
\newcommand{\rngi}[2]{RNG_{#1}^{#2}}
\newcommand{\mrd}[1][\minpoints]{mrd_{#1}}
\newcommand{\cd}[1][\minpoints]{c_{#1}}
\newcommand{\mrg}[1][\minpoints]{G_{#1}}

\newcommand{\maxk}{k_{max}}
\newcommand{\mptsrange}[1][\maxk]{k_1, \ldots, {#1}}
\newcommand{\mptssequence}[1][\mptsrange]{[{#1}]}
\newcommand{\mptsset}[1][\mptsrange]{\{{#1}\}}

\newcommand{\unfilteredRNG}{RNG**}
\newcommand{\smartRNG}{RNG*}
\newcommand{\exactRNG}{RNG}

\newcommand{\pple}{Property of the Largest Element }

\newcommand{\prng}[2]{P_{#1}(\rng{#2})}

\newcommand{\minpointsNN}[1]{#1}

\newcommand{\js}[1]{\textcolor{blue}{JS:$>>$} \textcolor{green}{#1} \textcolor{blue}{$<<$}}
\newcommand{\mn}[1]{\textcolor{blue}{MN:$>>$} \textcolor{red}{#1} \textcolor{blue}{$<<$}}
\newcommand{\ac}[1]{\textcolor{magenta}{AC:$>>$} \textcolor{cyan}{#1} \textcolor{magenta}{$<<$}}
\newcommand{\rc}[1]{\textcolor{red}{RC:$>>$} \textcolor{blue}{#1} \textcolor{red}{$<<$}}

\newtheorem{mydef}{Definition}
\newtheorem{corollary}{Corollary}
\newtheorem{theorem}{Theorem}
\newtheorem{proof}{Proof}

\renewcommand{\algorithmicrequire}{\textbf{Input:}}  
\renewcommand{\algorithmicensure}{\textbf{Output:}} 

\title{Efficient Computation of Multiple Density-Based Clustering Hierarchies}

\author{\IEEEauthorblockN{Antonio Cavalcante Araujo Neto\IEEEauthorrefmark{2},
Joerg Sander\IEEEauthorrefmark{2},
Ricardo J. G. B. Campello\IEEEauthorrefmark{3}, and
Mario A. Nascimento\IEEEauthorrefmark{2}}

\IEEEauthorblockA{\IEEEauthorrefmark{2}Department of Computing Science, University of Alberta, Canada}

\IEEEauthorblockA{\IEEEauthorrefmark{3}College of Science and Engineering, James Cook University, Australia}

\{antonio.cavalcante, jsander, mario.nascimento\}@ualberta.ca\IEEEauthorrefmark{2}, ricardo.campello@jcu.edu.au\IEEEauthorrefmark{3}
}






\maketitle

%
%
\begin{abstract}
\boldmath
HDBSCAN*, a state-of-the-art density-based hierarchical clustering method, produces a hierarchical organization of clusters in a dataset w.r.t. a parameter $\minpoints$. While the performance of HDBSCAN* is robust w.r.t. $\minpoints$ in the sense that a small change in $\minpoints$ typically leads to only a small or no change in the clustering structure, choosing a ``good'' $\minpoints$ value can be challenging: depending on the data distribution, a high or low value for $\minpoints$ may be more appropriate, and certain data clusters may reveal themselves at different values of $\minpoints$. 
To explore results for a range of $\minpoints$ values, however, one has to run HDBSCAN* for each value in the range independently, which is computationally inefficient.
In this paper, we propose an efficient approach to compute all HDBSCAN* hierarchies for a range of $\minpoints$ values by replacing the graph used by HDBSCAN* with a much smaller graph that is guaranteed to contain the required information.
An extensive experimental evaluation shows that with our approach one can obtain over one hundred hierarchies for the computational cost equivalent to running HDBSCAN* about 2 times.
\end{abstract}

%
%
\section{Introduction}
\label{sec:introduction}

The discovery of groups within datasets plays an important role in the exploration and analysis of data.
For scenarios where there is little to no prior knowledge about the data, clustering techniques are widely used. 
Density-based clustering, in particular, is a popular clustering paradigm that defines clusters as high-density regions in the data space, separated  by low-density regions. Algorithms in this class, such as DBSCAN \cite{DBLP:conf/kdd/EsterKSX96}, DENCLUE \cite{DBLP:conf/kdd/HinneburgK98}, OPTICS \cite{DBLP:conf/sigmod/AnkerstBKS99} and HDBSCAN* \cite{DBLP:journals/tkdd/CampelloMZS15}, stand out for their ability to find clusters of arbitrary shapes and to differentiate between cluster points and noise.

HDBSCAN*, the current state-of-the-art, computes a {\em hierarchy} of nested clusters, representing clusters at different density levels. It generalizes and improves several aspects of previous algorithms, and allows for a comprehensive framework for cluster analysis, visualization, and unsupervised outlier detection \cite{DBLP:journals/tkdd/CampelloMZS15}.
It requires a single parameter $\minpoints$, a smoothing factor that can implicitly influence which clusters are detectable in the cluster hierarchy.
Choosing a ``correct'' value for $\minpoints$ is typically not trivial. 
For instance, consider the examples in Figure \ref{fig:example}, which shows the results of  HDBSCAN* (with automatic cluster extraction) for two datasets $A$ and $B$ and two sample $\minpoints$ values, $\minpoints = 5$ and $25$, selected after running HDBSCAN* multiple times to both datasets w.r.t. $\minpoints \in \{2,3,...,100\}$. 
The main points here is that (1) there is no single value of $\minpoints$ that would result in the extraction of the clusters in both cases, and (2) a user would not know for a general dataset which value for $\minpoints$ is suitable. Dataset A is completely labeled as noise for $\minpoints > 24$, while the two structures in dataset B only start to be detected as separate clusters for $\minpoints > 24$. It may even be the case that different values of $\minpoints$ are needed to reveal clusters in different areas of the data space of a single dataset. 

\begin{figure*}[hbt] 
	\centering
	\subfloat[$A$, $\minpoints = 5$] {\includegraphics[width=.5\columnwidth,trim={0 0 0 1.2cm},clip]{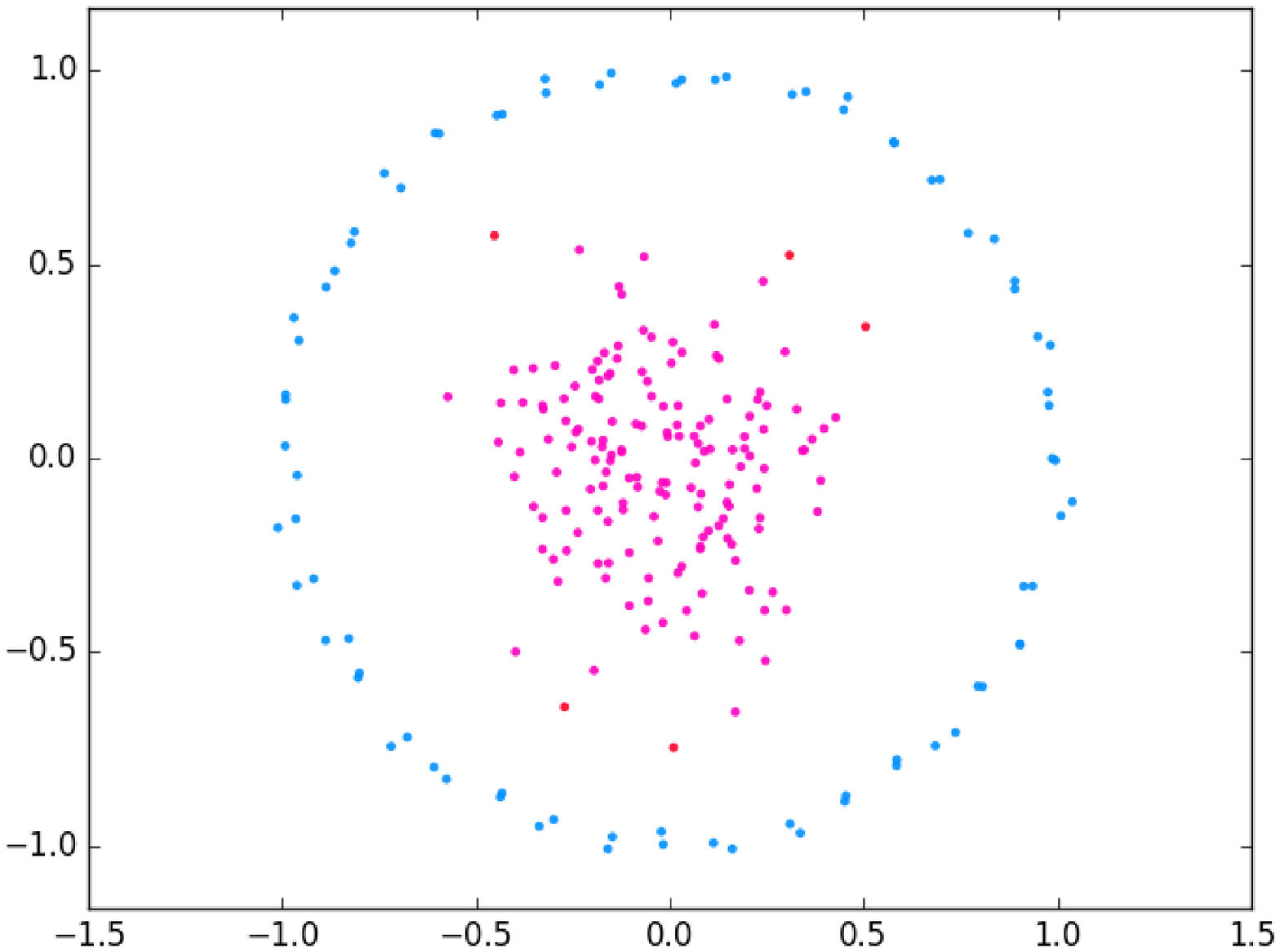}}
    \subfloat[$A$, $\minpoints = 25$] {\includegraphics[width=.5\columnwidth,trim={0 0 0 1.2cm},clip]{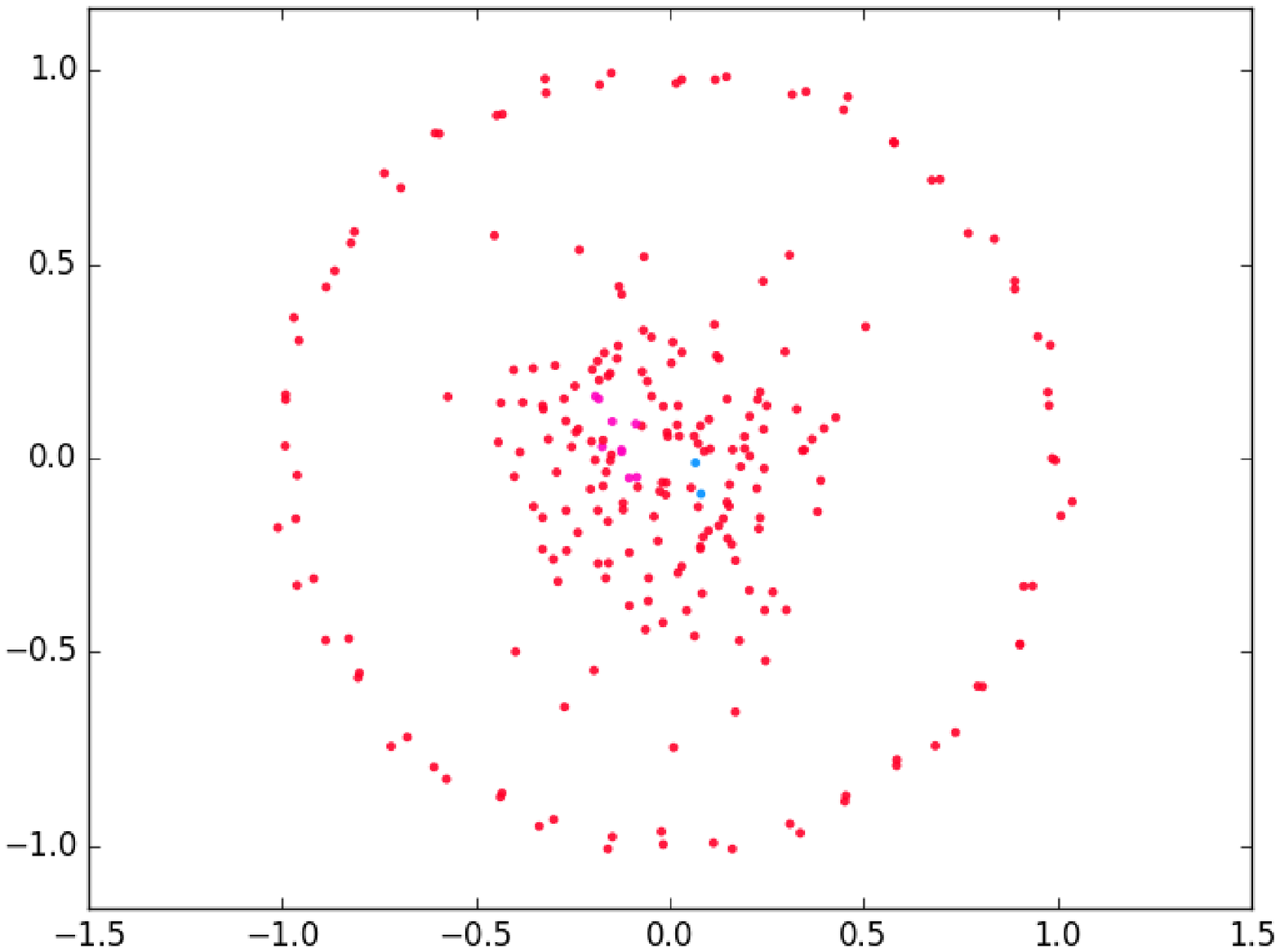}}
	\subfloat[$B$, $\minpoints = 5$]{\includegraphics[width=.5\columnwidth,trim={0 0 0 1.2cm},clip]{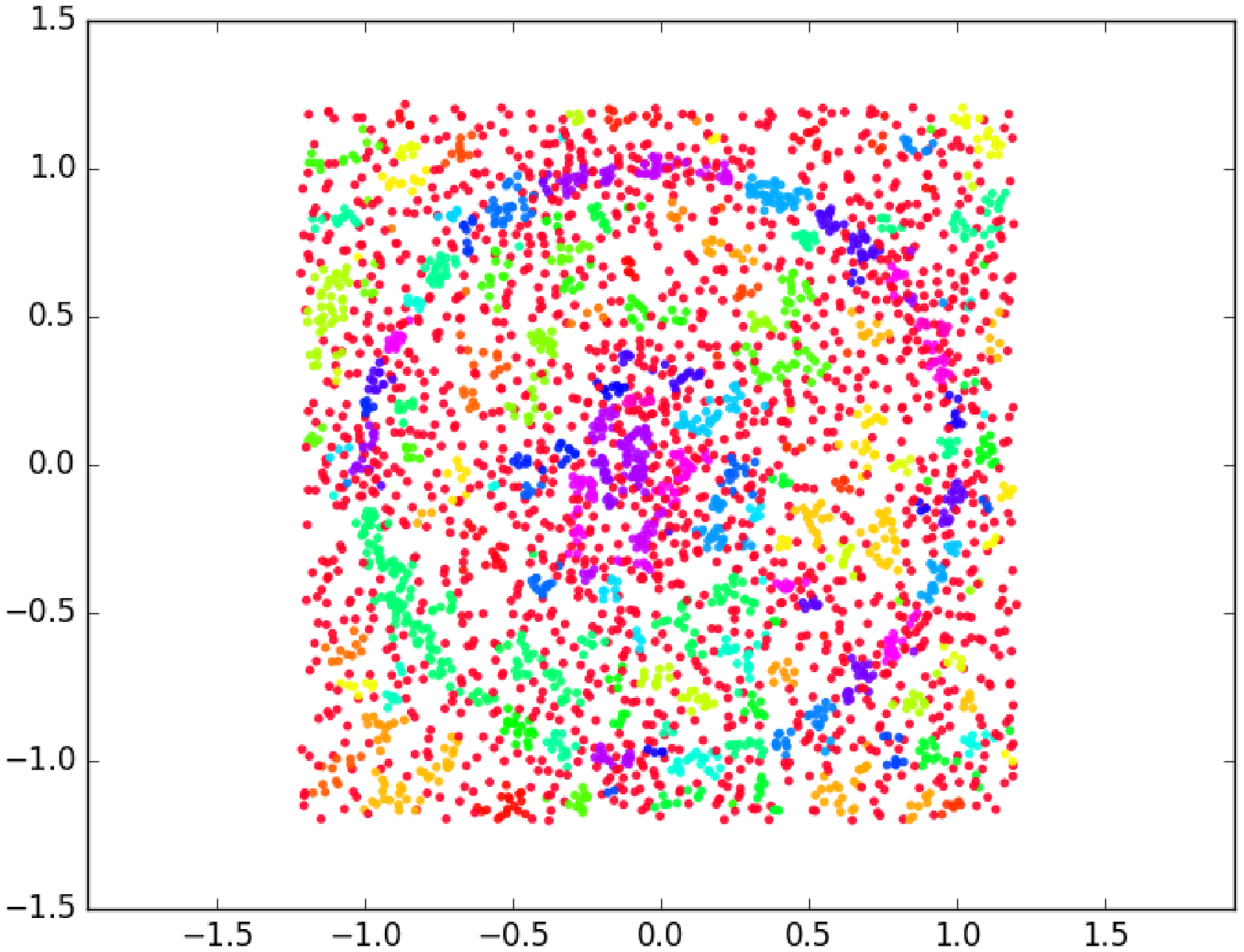}}
	\subfloat[$B$, $\minpoints = 25$] {\includegraphics[width=.5\columnwidth,trim={0 0 0 1.2cm},clip]{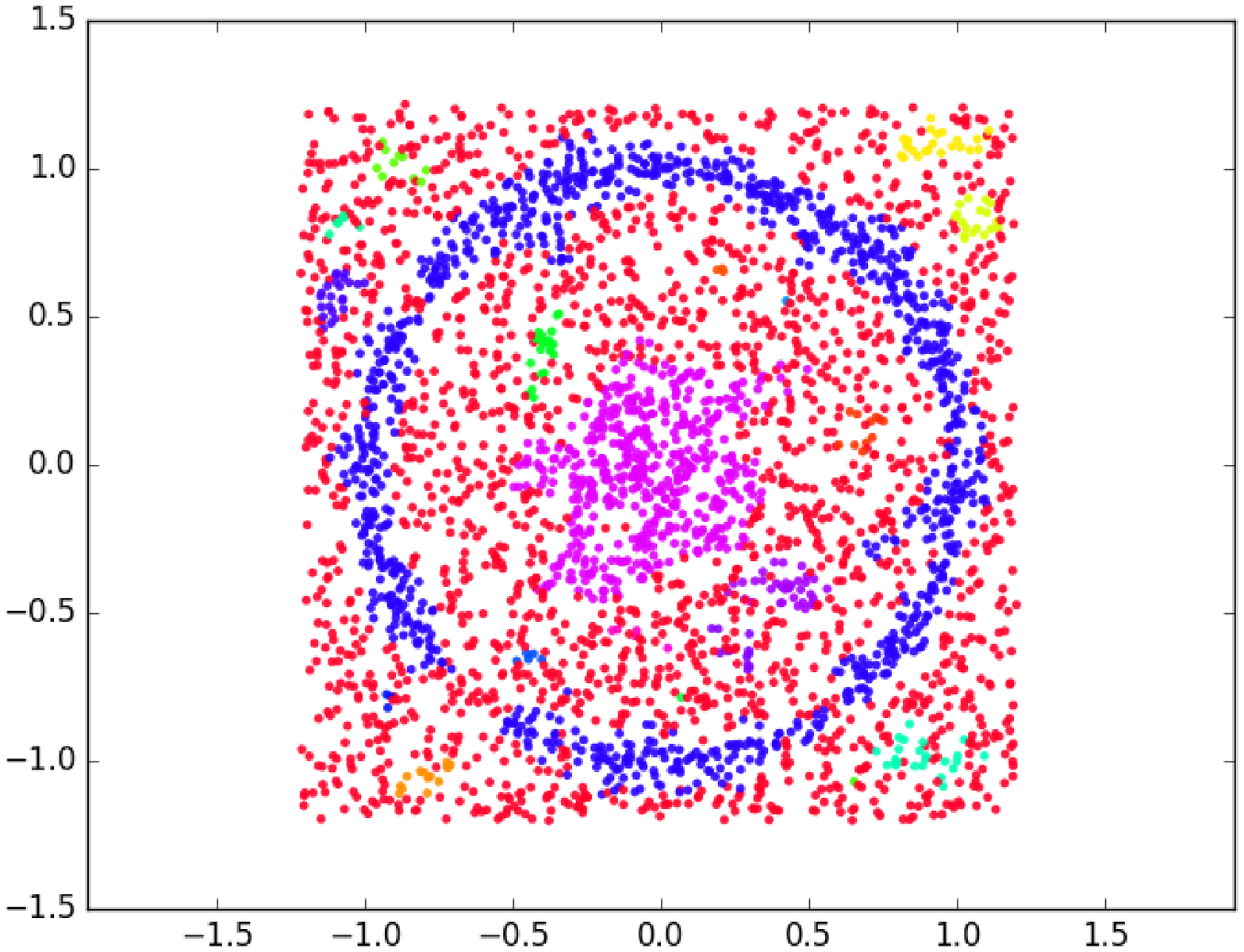}}
    \caption{Clusters from datasets $A$ and $B$ for $\minpoints = 5$ and $25$. Noise points are colored in red in all plots.} 
    \label{fig:example}
\end{figure*} 

To analyze clustering structures in practice, users typically run HDBSCAN* (like other algorithms with a parameter) multiple times with several different $\minpoints$ values, and explore the resulting hierarchies. 
Ideally, one would want to analyze cluster structures w.r.t. a large range of $\minpoints$ values, to fully explore a dataset in a given application.
A larger range of HDBSCAN* solutions for a multiple values of $\minpoints$ values offers greater insight into a dataset, also
providing additional opportunities for exploratory data analysis. For instance, using internal cluster validation measures such as DBCV \cite{DBLP:conf/sdm/MoulaviJCZS14}
, one can identify promising density levels from different hierarchies, produced from different tunings of the algorithm's density estimates (based on $\minpoints$).

However, one is typically constrained by the required runtime for running HDBSCAN* once for each desired value of $\minpoints$, resulting from the fact that HDBSCAN* is based on computing a Minimum Spanning Tree (MST) for a \emph{complete} graph, for a given value of $\minpoints$.
Even though this complete graph does not need to be explicitly stored, it has $O(n^2)$ edges (for $n$ data points) whose weights depend on $\minpoints$. For each desired value of $\minpoints$, these weights have to be re-computed and an MST has to be constructed for the corresponding complete graph. We note that the computational cost for the MST construction depends on the number of edges in the input graph, $O(n^2)$ in this case.

As the main contribution in this paper we provide theoretical and practical results that lead us to a method for computing multiple hierarchies w.r.t a range of $\minpoints$ values $(\mptsrange)$,
which is much more efficient than re-running HDBSCAN* for each $\minpoints$ in this range.
This gives access to a large range of HDBSCAN* solutions for a low computational cost, in fact equivalent to the cost of running HDBSCAN* for only 1 or 2 values of $\minpoints$. 
To achieve that, we show the following:
\begin{enumerate}
\item The smallest known neighborhood graph that contains the Euclidean Minimum Spanning Tree (EMST) is the relative neighborhood graph (RNG) --- as a first step towards finding a small, single spanning subgraph that can replace the complete graph in HDBSCAN*, while maintaining HDBSCAN*'s correctness.
\item The proximity measure used in HDBSCAN*, which depends on $\minpoints$, can be used to define RNGs that can replace the complete graph in HDBSCAN* with a corresponding RNG, one for each value of $\minpoints$.
\item For a range of $\minpoints$ values, RNGs w.r.t. smaller values are contained in RNGs w.r.t. larger values of $\minpoints$, so that a {\em single} RNG is sufficient to compute the hierarchies for the whole range of $\minpoints$ values.
\item Information (related to ``core-distances'') that is needed in HDBSCAN* and that can be computed in a pre-processing step, allows us to  formulate a highly efficient strategy for computing the single RNG, suitable for a whole range of $\minpoints$ values. 
\end{enumerate}
These results allow us to replace the (virtual) complete graph
of the data, on which HDBSCAN* is based, with a single, pre-computed RNG 
that contains all the edges needed to compute the hierarchies for every value of $\minpoints \in [1, \maxk]$. Moreover, this RNG has typically much fewer edges than 
the complete graph so its initial construction cost is more than outweighed by the reduction in edge weight computations.


The remainder of this paper is organized as follows.
Section \ref{sec:relatedwork} discusses related work. Section \ref{sec:background} covers basic concepts and techniques used in this paper. Section \ref{sec:approach} presents our proposal and proves  its correctness. Section \ref{sec:experiments} shows and discusses the results of our experimental evaluation. Section \ref{sec:conclusion} addresses the conclusions and some directions for future work.

%
%
\section{Related Work}
\label{sec:relatedwork}

To the best of our knowledge, there is no previous study of computing multiple clustering hierarchies efficiently. 
%
There has been work on automatic parameter selection strategies for density-based clustering, \emph{e.g.}, \cite{DBLP:conf/icebe/ChenMZW08,6249802,karami2014choosing}, which are loosely related to the issue illustrated in Figure~\ref{fig:example}. However, those proposals are unsuitable to be used with HDBSCAN*, since they were developed for non-hierarchical clustering algorithms.  In addition, they rely on assumptions that are often not satisfied in practice and there is not enough  evidence to support their claims about parameter optimality.

If we denote the HDBSCAN*'s (virtual) complete graph for a given $\minpoints$ by $\mrg$, a line of work related to our goal
of reducing the cost for computing an MST for each $\mrg$, 
are the works regarding (1) dynamically updating graphs, specifically MSTs, and (2)  neighborhood graphs that could potentially replace HDBSCAN*'s (virtual) complete graph.

%
For instance, the works of 
\cite{DBLP:journals/dam/CattaneoFPI10, DBLP:conf/icalp/HenzingerK97, DBLP:conf/soda/AmatoCI97}, studied the problem of maintaining dynamic MSTs. 
However, these approaches more suitable when the changes in the underlying graph 
take place sequentially, that is, considering each operation (\emph{e.g.}, edge updates) individually. 
When it comes to major changes taking place globally and simultaneously across the entire graph, as opposed to a few localized changes, a sequence of applications of these techniques tends to be computationally very costly, possibly even more costly than the construction of the entire MST from scratch. This is the case for $\mrg$, which is a complete graph whose majority of edges will likely change as a result of a change in the $\minpoints$ value.

%
The works on neighborhood graphs that are most related to our proposal aim at speeding up the secial case of computing a \emph{Euclidean} Minimum Spanning Tree (EMST), by first computing a spanning subgraph that is guaranteed to contain all the EMST edges. One of these strategies uses a Delaunay Triangulation \cite{citeulike:12792239} of the complete Euclidean graph $G$, since it has been shown that the EMST is contained in the Delaunay Triangulation of $G$ \cite{DBLP:journals/pr/Toussaint80}. Other spanning subgraphs of the complete graph $G$ that contain the EMST are the Gabriel Graph \cite{Gabriel01091969, citeulike:3982779} and the Relative Neighborhood Graph (RNG) \cite{DBLP:journals/pr/Toussaint80}. 
Unfortunately, these results are not simply applicable to our problem because $\mrg$ lies in a transformed space of the data that depends on $\minpoints$ ($\mrg \neq G$), and it is one of the main contributions of this paper to formally show how to adapt an RNG so that it can be use by HDBSCAN* as a suitable replacement for $\mrg$ for different $\minpoints$.

%
%
\section{Background}
\label{sec:background}


\subsection{HDBSCAN*}
\label{sec:hdbscan*}
HDBSCAN* is a hierarchical, density-based clustering algorithm that improves on previous density-based algorithms \cite{DBLP:conf/pakdd/CampelloMS13}. 
Its main output is a cluster hierarchy that describes the nested structure of density-based clusters in a dataset with respect to a single parameter, $\minpoints$, which can be seen as a smoothing factor that can affect how and if certain clusters are represented in the hierarchy. A specific level at some ``height'' $\varepsilon$ in this hierarchy represents a density level, specified by $\varepsilon$ and $\minpoints$, in terms of density-based clusters and noise, defined as follows: i) a point is either a \emph{core point} w.r.t. $\varepsilon$ and $\minpoints$ \emph{iff} it has at least $\minpoints$ many points in its $\varepsilon$-neighborhood, or a \emph{noise point} otherwise; ii) two core points are (directly) \emph{$\varepsilon$-reachable} w.r.t. $\minpoints$ if they are within each other's $\varepsilon$-neighborhood; iii) two core points are \emph{density-connected} w.r.t. $\varepsilon$ and $\minpoints$ if they are directly or transitively $\varepsilon$-reachable w.r.t. $\minpoints$; and iv) a \emph{cluster} w.r.t. $\varepsilon$ and $\minpoints$ is a (non-empty) maximal subset of density connected points w.r.t. $\varepsilon$ and $\minpoints$.
To determine the \emph{nested} structure of density-based clusters in a dataset $\textbf{X}$, w.r.t. $\minpoints$, one needs to know (i) for each point $p \in \textbf{X}$: the smallest value of $\varepsilon$ such that $p$ is a core point w.r.t. $\varepsilon$ and $\minpoints$, called $p$'s ``core-distance''
w.r.t. $\minpoints$; and (ii) for each value of $\varepsilon$: the clusters and the noise w.r.t. $\varepsilon$ and $\minpoints$.
The latter information can be derived conceptually from a complete, edge-weighted graph where nodes represent the points in $\textbf{X}$, and the \emph{edge weight} of an edge between two points $p$ and $q$ --- called the ``mutual reachability distance'' (w.r.t. $\minpoints$) between $p$ and $q$ --- is the smallest value of $\varepsilon$ such that $p$ and $q$ are (directly) $\varepsilon$-reachable w.r.t. $\minpoints$. This graph is called the ``Mutual Reachability Graph'', $\mrg$, which forms the conceptual basis of HDBSCAN*, and which we will discuss in more technical detail in the following subsection. For a specific density level ($\varepsilon$ and $\minpoints$), removing all edges from $\mrg$ with weights greater than $\varepsilon$ reveals the maximal, connected components, i.e., clusters, of that density level. The density-based clustering hierarchy can thus be compactly represented by (and more easily be extracted from) a Minimum Spanning Tree (MST) of $\mrg$. 

The HDBSCAN* hierarchy w.r.t. $\minpoints$ for a dataset \textbf{X} is computed in the following way:
First, the core distances of all points in \textbf{X} w.r.t. $\minpoints$ are computed.
Then, an MST of $\mrg$ is dynamically computed (without materializing $\mrg$).
From this MST, the \emph{complete} density-based cluster hierarchy w.r.t. $\minpoints$ is then extracted, by removing edges from the MST in descending order of edge weight, and \mbox{(re-)labeling} the connected components and noise at the resulting ``next'' level. 

\subsection{Mutual Reachability Graphs}
\label{sec:mrg}
The $\mrg$ is a complete, edge-weighted graph that represents the mutual reachability relationship between any two points in a dataset \textbf{X}. Each point in \textbf{X} corresponds to a vertex, and between each pair of points $p$ and $q$, there is an edge whose weight is defined as the Mutual Reachability Distance between $p$ and $q$ w.r.t. $\minpoints$, $\mrd$ \cite{DBLP:conf/icdm/LelisS09}:
\begin{equation} \label{eq:mrd}
\mrd(p, q) = max\{\cd(p), \cd(q), d(p, q)\}
\end{equation}
where $d(\cdot, \cdot)$ represents the underlying distance function (typically Euclidean distance), and $\cd(p)$ represents the core distance of $p$, which is formally the distance from $p$ to its $\minpoints$ nearest neighbor, $\minpoints\operatorname{-NN}(p)$:
\begin{equation}
\cd(p) = d(p, \minpoints\operatorname{-NN}(p))
\end{equation}

\noindent In this work, we assume that the underlying distance $d(\cdot, \cdot)$ satisfies Symmetry and Triangle Inequality, and, without loss of generality, we use Euclidean Distance in our examples.

Intuitively, an edge weight in $\mrg$ corresponds to the minimum radius $\varepsilon$ at which the corresponding endpoints are directly $\varepsilon$-reachable w.r.t. $\minpoints$, \emph{i.e.}, the smallest distance at which both points are in each other's $\varepsilon$-neighborhood, and both $\varepsilon$-neighborhoods contain at least $\minpoints$ points.

The Mutual Reachability Graph has the following important characteristics related to $\mrd$ and to how these edge weights change when changing the value of $\minpoints$: 1) Increasing the value of $\minpoints$ leads, in general, to higher values of $\cd$, since $\cd$ is the $\minpoints$-th nearest neighbor distance; 2) When increasing the value of $\cd$, more edges will have the same edge weight, since a point $p$ with a high $\cd$ determines the weight of all edges between $p$ and its $\minpoints$-nearest neighbors that have a smaller $\cd$ than $p$ (given that $\mrd$ is defined by a max function); 3) When decreasing the value of $\minpoints$, edge weights can either decrease or remain the same, but never increase.

The authors of HDBSCAN* \cite{DBLP:conf/pakdd/CampelloMS13}, \cite{DBLP:journals/tkdd/CampelloMZS15} deem $\mrg$ a conceptual graph as it does not need to be explicitly stored; edge weights can be computed on demand, when needed.  

%
%
\section{Approach}
\label{sec:approach}

At the core of HDBSCAN* is the computation of an MST from the $\mrg$ of a dataset.
The time needed to compute an MST depends on the number of edges of the input graph, which in case of $\mrg$ is a complete graph. Even if $\mrg$ is not materialized, $O(n^2)$ edge weights have to be processed for a dataset with $n$ points. 

When HDBSCAN* has to be run for a range, $\mptsrange$, of $\minpoints$ values, many MSTs have to be computed for different $\mrg$ graphs, one for each value of $\minpoints \in \mptsset$. Conceptually, we can think of each $\mrg$ being obtained by taking the complete, unweighted graph $G$ of the dataset, and then incorporate into it appropriate edge weights, which means that edge weights of $\mrg$ have to be re-computed for each $\minpoints \in \mptsset$ before an MST is  constructed.

One way to speed up the execution time of HDBSCAN* over all values of $\minpoints \in \mptsset$, even with a naive approach of re-running HDBSCAN* for each $\minpoints$ value, is to execute $k$-Nearest-Neigbor ($k$-NN) queries for each point only once, using the largest value $\maxk$ in the range. When computing the core distance of a point $p$ w.r.t. $\minpoints = \maxk$ using a $\maxk$-NN query, the information about \emph{all smaller} core distances of $p$ (\emph{i.e.}, w.r.t. $\minpoints = k_j$, where $j \leq max$), is readily available as part of the $\maxk$-NN query computation. Hence, the core distances for all values of $\minpoints \in \mptsset$ can be pre-computed and stored so that the re-computation of edge weights (reachability distances) for the different $\mrg$ graphs does not require additional $k$-NN-queries.
However, even with pre-computed core distances, a major factor determining the total runtime of HDBSCAN*, over all values of $\minpoints \in \mptsset$, is the large number of edges that have to be processed in the MST constructions
for each value of $\minpoints$.

To reduce the number of edges that have to be processed, we can ask the questions: is it possible to construct a single graph that is significantly smaller than a complete graph, and that contains all the edges needed to compute the MST of $\mrg$ for all $\minpoints \in \mptsset$? If the answer is \emph{yes}, we can use this graph instead of the complete graph in HDBSCAN*, without changing the correctness of the result: we can just re-compute its edge weights instead of the edge weights of the complete graph, for each value of $\minpoints \in \mptsset$, and compute the MST of \emph{this} edge-weighted graph. In the following, we will formally prove that one of the known graphs can be adapted so that it can be used in our approach to running HDBSCAN* for each value of $\minpoints \in \mptsset$. How much speed-up can be achieved in this manner depends, however, not only on the reduction in number of edges from the complete graph, but also on the added computational cost for constructing this graph. 

\subsection{Results from Computational Geometry}

Consider first the special case of $\minpoints = 1$, where all core distances are equal to zero, and thus the mutual reachability distance $\mrd$ reduces to the underlying distance function. With Euclidean distance, what HDBSCAN* has to compute then is the Euclidean Minimum Spanning Tree (EMST) of a dataset $\textbf{X}$, i.e., the MST of a complete graph of $\textbf{X}$ (containing an edge between every pair of points/vertexes) with Euclidean distance between points as edge weights. 

For the EMST, there are known results from computational geometry that relate the EMST to some of the so-called \emph{proximity graphs}, in which two points are connected by an edge whenever a certain spatial constraint is satisfied. The most popular ones are the Delaunay Triangulation (DT), the Gabriel Graph (GG) and the Relative Neighborhood Graph (RNG), for which it has been shown that \cite{DBLP:journals/pr/Toussaint80}:
\begin{equation}
EMST \subseteq RNG \subseteq GG \subseteq DT \label{eq:dt-gg-rng-emst}
\end{equation}
The RNG and GG are special cases of a family of graphs called $\beta$-skeletons \cite{kirkpatrick1984framework}, which can range from the complete graph to the empty graph, when $\beta$ goes from 0 to $\infty$.
A value of $\beta = 1$ results in the GG and $\beta = 2$ results in the RNG.

Given this result, the RNG, or possibly a $\beta$-skeleton with even fewer edges, may be a good replacement for a complete graph, if we can answer the following questions positively:
\begin{enumerate}
\item Can we determine the smallest $\beta$-skeleton, in terms of number of edges, that contains the EMST as a subgraph?
\item Can the results we have for Euclidean distance
be generalized to other reachability distances w.r.t. $\minpoints > 1$? 
\item Is there a single $\beta$-skeleton that contains all the edges needed to compute an MST of $\mrg$ for each value of $\minpoints$ in a range of values $\mptsrange$?
\item Does the reduction in the number of edges justify the additional computational cost for constructing and materializing a $\beta$-skeleton for our task? 
\end{enumerate}
We will answer these questions in the following subsections.

\subsection{The Smallest $\beta$-Skeleton Containing the EMST}
\label{sec:RNG}
The family of $\beta$-skeletons for a set of $d$-dimensional points is defined in the following way. 
For a given $\beta$, an edge exists between two points $a$ and $b$ if 
the intersection of the two balls centered at $(\beta/2)a + (1- \beta/2)b$ and $(1-\beta/2)a+(\beta/2)b$, both with radius $\beta d(a,b)/2$, is empty.
For instance, when $\beta = 2$ (the case of an RNG), the centers of the two balls coincide with the points $a$ and $b$, and their radius is equal to $d(a, b)$, as illustrated in Figure \ref{fig:beta2}. The highlighted region, called $lune(a,b)$, must be empty for $a$ and $b$ to be connected via an edge. 
For $\beta = 2$, one can equivalently say that an edge exists between $a$ and $b$ if 
\begin{equation}
d(a, b) \leq max\{d(a, c), d(b, c)\}, \forall c \ne a, b
\label{eq:rng_edge}
\end{equation}



	
				
	
	



The RNG is guaranteed to contain the EMST, which has been shown in \cite{DBLP:journals/pr/Toussaint80}.
The essence of the proof can be demonstrated considering a configuration of three points $a$, $b$, $c$, such that $lune(a,b)$ contains $c$, as shown in Figure \ref{fig:k}.
The edges $(a, b)$, $(a, c)$ and $(b, c)$ cannot all be part of an EMST, as they form a cycle.
Since $(a, b)$ is the largest of these edges, $(a, b)$ cannot be part of the EMST. Thus, a necessary (yet not sufficient) condition for an edge $(a,b)$ to be in an EMST is that all other points must lie outside $lune(a,b)$. In other words, we can remove all edges $(a,b) \in E$ from the complete graph $G$ whose points $a$ and $b$ have a non-empty lune, and be sure that the resulting graph $(V, E \setminus \{(a,b): lune(a,b) \neq \emptyset\}) = RNG$ contains the EMST.
Here, we prove by counter example that the RNG is actually the smallest $\beta$-skeleton graph (i.e., there is no $\beta > 2$) with that property. 
Consider a dataset with three points $a$, $b$ and $c$, located at equal distance from each other, as illustrated in Figure \ref{fig:beta}:
When $\beta = 2$ (Figure \ref{fig:beta2}), according to Inequality (\ref{eq:rng_edge}), there is an edge between every pair of points in the $2$-skeleton of this dataset. 
For any $\beta > 2$ (Figure \ref{fig:beta3}), however, the radius of the balls that define $lune(a, b)$ is increased by a factor of $(\beta-2)/2$, and the centers of the balls  are ``pulled apart'' accordingly, so that $c$ (equidistant from $a$ and $b$) must now be inside $lune(a,b)$. Thus, $a$ and $b$ are (by definition of $\beta$-skeleton) no longer connected by an edge.
Analogously, there is no edge between the other pairs of points for $\beta > 2$, resulting in an empty $\beta$-skeleton that obviously cannot contain the EMST (i.e., $EMST \nsubseteq$ $\beta$-skeleton $\forall \beta > 2$). From this result and from the known result in Expression (\ref{eq:dt-gg-rng-emst}), the RNG ($\beta = 2$) is thus the smallest $\beta$-skeleton graph that contains the EMST and, for this reason, we choose it as the basis for further analysis. 

\begin{figure}[t!]
	\centering
    \adjustbox{valign=b}{
	\subfloat[$\beta = 2$]{%
		\centering
		\def\A{(1,1) circle (1)}
		\def\B{(2,1) circle (1)}		
		\begin{tikzpicture}
			\draw[thick,dashed] \A;
			\draw[thick,dashed] \B;	
			\fill (1,1) circle[radius=2pt] node[align=left, right] {$a$};
			\fill (2,1) circle[radius=2pt] node[align=right, left] {$b$};	
			\fill (1.5,1.86) circle[radius=2pt] node[align=left, below, inner sep=4pt] {$c$};
           	\begin{scope}
				\clip \A;
				\fill[pattern=north east lines] \B;
			\end{scope}

			\draw (2,-0.25) -- (2,-0.25);
		\end{tikzpicture}    
		\label{fig:beta2}
    }}
    \adjustbox{valign=b}{
    \subfloat[$\beta > 2$]{%
		\centering
		\def\C{(4.75,1) circle (1.25)}
		\def\D{(6.25,1) circle (1.25)}
		\begin{tikzpicture}
			\draw[thick,dashed] \C;
			\draw[thick,dashed] \D;
			\fill (4.75,1) circle[radius=0.6pt] node[align=left, right] {};
			\fill (6.25,1) circle[radius=0.6pt] node[align=left, right] {};
			\fill (5,1) circle[radius=2pt] node[align=left, right] {$a$};
			\fill (6,1) circle[radius=2pt] node[align=right, left] {$b$};	
			\fill (5.5,1.86) circle[radius=2pt] node[align=right, below, inner sep=4pt] {$c$};
			\draw (5,-0.25) -- (5,-0.25);
		\end{tikzpicture}
		\label{fig:beta3}
	}}
	\caption{$\beta$-skeletons}
	\label{fig:beta}
\end{figure}
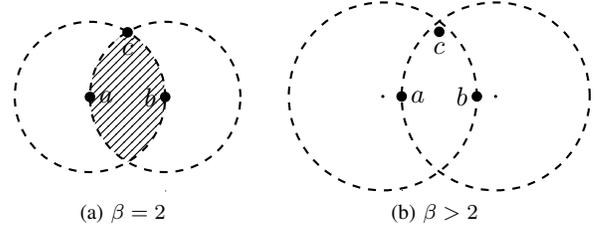

\subsection{The RNG w.r.t. Mutual Reachability Distance}
In this section, we prove that the results for RNGs in Euclidean space can be extended to the space of mutual reachability distances. 

\emph{Notation:} (1) Let $G = (V, E)$ denote the undirected, unweighted complete graph corresponding to a dataset, \emph{i.e.}, the set of vertexes $V$ represents the data points, and the set of edges $E \subset V \times V$ represents all pairs of vertexes/points. (2) Let $\mrg[i] = (V, E, \mrd[i])$ be the mutual reachability graph $\mrg$ for $\minpoints = i$, \emph{i.e.}, the weighted, complete graph for the dataset with edge weights equal to $\mrd[i]$, the mutual reachability distance w.r.t. $\minpoints = i$.

We can define a relative neighborhood graph w.r.t. the mutual reachability distance $\mrd[i]$, $\rng{i}$, as follows:

\begin{mydef}
$\rng{i} = (V, E')$ where $E' \subseteq E$ and there is an edge $(a,b) \in E'$ if and only if:
\begin{align*}
\mrd[i](a, b) \leq max\{\mrd[i](a, c), \mrd[i](b, c)\}, \forall c \ne a, b;
\end{align*} 
and when there is an edge $(a,b) \in E'$, we say that $a$ and $b$ are relative neighbors w.r.t $\mrd[i]$.
The unweighted graph $\rng{i}$ can be extended with edge weights defined by a distance function $\mrd[j]$, which results in the edge weighed graph $\rngi{j}{i}$, where the weight of an edge connecting two points $p$ and $q$ is equal to $\mrd[j](p,q)$. 
\label{def:RNG}
\end{mydef}

We can prove that the $\rngi{i}{i}$ contains the MST of $\mrg[i]$, and thus we can replace $\mrg[i]$ with $\rngi{i}{i}$ when running HDBSCAN* for $\minpoints = i$.

\begin{theorem}
	$MST(\mrg[i]) \subseteq \rngi{i}{i}$
	\label{theo:1}
\end{theorem}

	
	
	
		

	
	
		
	

\begin{proof}
The argument for why $EMST \subseteq RNG$, which has been shown in \cite{DBLP:journals/pr/Toussaint80}, relies in essence only on the fact that the Euclidean Distance is symmetric and satisfies the triangle inequality; it is, in fact valid for any distance function with these properties, which are needed to guarantee that $(a, b)$ is in fact the largest edge in configurations like the one shown in Figure \ref{fig:k}. 
Consequently, we only need to show that $\mrd[i]$ satisfies symmetry and triangle inequality.

For symmetry, we can easily see from the definition of $\mrd$ in Equation (\ref{eq:mrd}) that $\mrd[i](a,b) = \mrd[i](b,a)$ (given that the underlying distance $d$ is symmetric by assumption). 

For the triangle inequality, we have to show that for all $a,b,c$ in a dataset $X$:
\begin{align}
	\mrd[i](a,c) \leq \mrd[i](a,b) + \mrd[i](b,c) \label{eq:inequality}
\end{align}

By assumption (Section \ref{sec:background}), the underling distance $d$ in the definition of $\mrd[i]$ satisfies the triangle inequality, \emph{i.e.}:
\begin{align}
	d(a,c) \leq d(a,b) + d(b,c) \label{eq:inequality_euclidean}
\end{align}

There are three cases according to the definition of $\mrd[i](a,c)$, in all of which the triangle inequality must hold:

1) $\mrd[i](a,c) = \cd[i](a)$. 
The $max$ function in the definition of $\mrd[i]$ implies $\cd[i](a) \leq \mrd[i](a,b)$. 
Hence, it follows that $\mrd[i](a,c) = \cd[i](a) \leq \mrd[i](a,b) \leq \mrd[i](a,b) + \mrd[i](b,c)$.

2) $\mrd[i](a,c) = \cd[i](c)$. 
Analogous to case 1).

3) $\mrd[i](a,c) = d(a,c)$. 
Since for any $x, y$ it holds that $x \leq max(x, y)$, we can replace the terms on the right side of Inequality (\ref{eq:inequality_euclidean}) with $max$ functions to obtain, $d(a,c) \leq max\{d(a,b), \cd[i](a), \cd[i](b)\} + max\{d(b,c), \cd[i](b), \cd[i](c)\} = \mrd[i](a,b) + \mrd[i](b,c)$, and hence, also in this case: $\mrd[i](a,c) = d(a,c) \leq \mrd[i](a,b) + \mrd[i](b,c)$.

Since $\mrd[i]$ satisfies symmetry and triangle inequality, it follows from \cite{DBLP:journals/pr/Toussaint80} that $\rngi{i}{i}$ contains the MST of $\mrg[i]$.
\end{proof} 

\subsection{One RNG To Rule Them All}

We have established that we can use $\rngi{i}{i}$ as a substitute for $\mrg[i]$ in HDBSCAN*.  
We will now show that all MSTs for HDBSCAN* hierarchies w.r.t. $\minpoints \in \mptsset$ can be obtained from the single graph $\rng{\maxk}$. For this we only need to show that $\rng{i} \subseteq \rng{\maxk}$, for all $i < \maxk$. If this property holds, we can use the single graph $\rng{\maxk}$ to compute the MST of any $\mrg[i]$ by first equipping $\rng{\maxk}$ with edge weights $\mrd[i]$, and then computing the MST of this edge-weighted graph $\rngi{i}{\maxk}$. Compared to the naive approach that uses the complete graph $G$ in this manner, we should be able to speed up the MST computations by using a graph that has typically much fewer edges.


\begin{theorem}
    $\rng{i} \subseteq \rng{\maxk}$, $\forall i < \maxk$.
	\label{theo:2}
\end{theorem}

\begin{proof}
To prove this by contradiction, assume that there is a $j < \maxk$ for which this property does not hold, \emph{i.e.}, $\rng{j} \not\subseteq \rng{\maxk}$.
Then, there must be at least one edge $(a,b)$ that belongs to $\rng{j}$ but does not belong to $\rng{\maxk}$. 
According to the condition that defines relative neighborhood graphs, this means that there is a point $c$, such that for distance $\mrd[\maxk]$, $c \in lune(a, b)$, and for distance $\mrd[j]$, $c \notin lune(a,b)$, as illustrated in Figure \ref{fig:proof}. 

\begin{figure}[!tb]
	\centering
    \adjustbox{valign=b}{
   	\subfloat[]{
		\centering
			
		\def\A{(1,1) circle (1)}
		\def\B{(2,1) circle (1)}
						
		\begin{tikzpicture}
			
		\draw[thick,dashed] \A;
		\draw[thick,dashed] \B;
						
		\fill (1,1) circle[radius=3pt] node[align=left, left] {$a$};
		\fill (2,1) circle[radius=3pt] node[align=right, right] {$b$};
						
		\fill (1.65,1.4) circle[radius=3pt] node[align=left, left] {$c$};
						
		\end{tikzpicture}
			
		\label{fig:k}
	}}
    \adjustbox{valign=b}{
	\subfloat[]{
		\centering
				
		\def\C{(5,1) circle (1)}
		\def\D{(6,1) circle (1)}
		
		\begin{tikzpicture}
				
		\draw[thick,dashed] \C;
		\draw[thick,dashed] \D;
				
		\fill (5,1) circle[radius=3pt] node[align=left, left] {$a$};
		\fill (6,1) circle[radius=3pt] node[align=right, right] {$b$};
		
		\fill (6,1.7) circle[radius=3pt] node[align=right, right] {$c$};
		
		\draw (5,1) -- (6,1);
		
		\end{tikzpicture}
		
		\label{fig:k-1}
	}}

	\caption{Illustration for proofs of Theorem \ref{theo:1} and \ref{theo:2}}
	\label{fig:proof}
\end{figure}
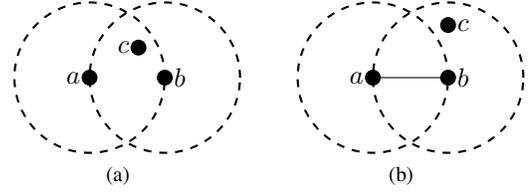

For $\rng{\maxk}$ (Figure \ref{fig:k}) this means that the following inequalities must \emph{both} be satisfied so that $c \in lune(a,b)$.
\begin{align}
	\mrd[\maxk](a,b) > \mrd[\maxk](a,c) \label{eq:1} \\
	\mrd[\maxk](a,b) > \mrd[\maxk](b,c) \label{eq:2}
\end{align}

For $\rng{j}$ (Figure \ref{fig:k-1}) this means that \emph{at least one} of the following inequalities must be satisfied so that $c \notin lune(a,b)$.
\begin{align}
	\mrd[j](a,b) \leq \mrd[j](a,c) \label{eq:3} \\
	\mrd[j](a,b) \leq \mrd[j](b,c) \label{eq:4}
\end{align}

Using the definition of $\mrd[\maxk]$, we can rewrite Inequalities \eqref{eq:1} and \eqref{eq:2} as follows:
\begin{align}
\begin{split}
		max\{\cd[\maxk](a), \cd[\maxk]&(b), d(a,b)\} \\
        &> \\
        max\{\cd[\maxk](a), \cd[\maxk]&(c), d(a,c)\} \label{eq:5} \\
\end{split}
\end{align}

\begin{align}
\begin{split}
		max\{\cd[\maxk](a), \cd[\maxk]&(b), d(a,b)\} \\ 
        &> \\
        max\{\cd[\maxk](b), \cd[\maxk]&(c), d(b,c)\} \label{eq:6}
\end{split}
\end{align}

There are theoretically three cases, $\cd[\maxk](a)$, $\cd[\maxk](b)$, and $d(a,b)$, that the $max$ function on the left-hand side of the Inequalities \eqref{eq:5} and \eqref{eq:6} can evaluate to. However, for both inequalities to be true simultaneously, only $d(a,b)$ is possible:

1) Case $max\{\cd[\maxk](a), \cd[\maxk](b), d(a,b)\} = \cd[\maxk](a)$. \\
In this case, we get from Inequality \eqref{eq:5} the following:
\begin{align*}
	\cd[\maxk](a) > max\{\cd[\maxk](a), \cd[\maxk](c), d(a,c)\}
\end{align*}
But since $max(\cd[\maxk](a), \ldots) \geq \cd[\maxk](a)$, it follows that $\cd[\maxk](a) > \cd[\maxk](a)$, a contradiction!

2) Case $max\{\cd[\maxk](a), \cd[\maxk](b), d(a,b)\} = \cd[\maxk](b)$. \\
In this case, analogously to the previous case, if follows from Inequality \eqref{eq:6} that $\cd[\maxk](b) > \cd[\maxk](b)$, a contradiction!
	
3) Case $max\{\cd[\maxk](a), \cd[\maxk](b), d(a,b)\} = d(a,b)$. \\
In this case there is no contradiction and, therefore, it is the only option to satisfy both \eqref{eq:5} and \eqref{eq:6} simultaneously.

Having established that the left-hand side of Inequalities \eqref{eq:5} and \eqref{eq:6} must be equal to $d(a,b)$, we can infer that all the following inequalities must hold.
\begin{align}
	d(a,b) &> \cd[\maxk](a) \label{eq:7} \\
    d(a,b) &> \cd[\maxk](b) \label{eq:8}  \\
	d(a,b) &> \cd[\maxk](c) \label{eq:9} \\
	d(a,b) &> d(a,c)   \label{eq:10} \\
	d(a,b) &> d(b,c)   \label{eq:11}
\end{align}

Let us now turn to the Inequalities \eqref{eq:3} and \eqref{eq:4} of which at least one must also hold, under our assumption that $c \notin lune(a,b)$ for distance $\mrd[j]$.  
We can rewrite \eqref{eq:3}, using the definition of $\mrd[j]$ as follows:
\begin{align}
  \begin{split}
          max\{\cd[j](a), \cd[j]&(c), d(a,c)\} \\ 
          &\geq \\
          max\{\cd[j](a), \cd[j]&(b), d(a,b)\} \label{eq:13}
  \end{split}
\end{align}

There are again three possible cases, $\cd[j](a), \cd[j](c), d(a,c)$, that the $max$ function on the left-hand side of Inequality \eqref{eq:13} can evaluate to, and we show that each one leads to a contradiction to what we already know about $a$, $b$, and $c$:

1) $max\{\cd[j](a), \cd[j](c), d(a,c)\} = \cd[j](a)$. \\
In this case, Inequality \eqref{eq:13} yields the following.
		\begin{align}
			\cd[j](a) &\geq d(a,b) \label{eq:15}
		\end{align}		
Since core distances $\cd$ can only increase when $\minpoints$ increases, we have $\cd[\maxk](a) \geq \cd[j](a)$ and, accordingly, we obtain the following from Inequality \eqref{eq:15}.
		\begin{align}
			\cd[\maxk](a) &\geq d(a,b) \label{eq:17}
		\end{align}
This contradicts Inequality \eqref{eq:7}!

2) $max\{\cd[j](a), \cd[j](c), d(a,c)\} = \cd[j](c)$. \\ Analogously to the previous case, from \eqref{eq:13} we get \eqref{eq:19}, and then from $\cd[\maxk](c) \geq \cd[j](c)$ we get \eqref{eq:21}, which contradicts Inequality \eqref{eq:9}!
\begin{align}
	\cd[j](c) &\geq d(a,b) 		\label{eq:19} \\
	\cd[\maxk](c) &\geq d(a,b) 	\label{eq:21}		
\end{align}

3) $max\{\cd[j](a), \cd[j](c), d(a,c)\} = d(a,c)$. \\
In this case, we get from Inequality \eqref{eq:13} that $d(a,c) \geq d(a,b)$, which is a contradiction to Inequality \eqref{eq:10}!

This proves that Inequality \eqref{eq:3} cannot hold under our assumption. We can prove analogously the same result for Inequality \eqref{eq:4}, which contradicts our assumption that there is a $j < \maxk$ such that $\rng{j} \not\subseteq \rng{\maxk}$.
Hence $\rng{i} \subseteq \rng{\maxk}$, $\forall i \leq \maxk$.
\end{proof}

When we combine the results of Theorems \ref{theo:1} and \ref{theo:2}, we obtain the following corollary, which states that the $MST(\mrg[i])$ for all $i < \maxk$ is contained in $\rng{\maxk}$, and can thus be obtained by extending $\rng{\maxk}$ with edge weights $\mrd[i]$, and computing the MST of this graph $\rngi{i}{\maxk}$.

\begin{corollary}
$MST(\mrg[i]) \subseteq \rngi{i}{\maxk}$, $\forall i \leq \maxk$.
\end{corollary}

\begin{proof}
$MST(\mrg[i]) \subseteq \rngi{i}{i}$ (Theorem \ref{theo:1}), and $\rng{i} \subseteq \rng{\maxk}$ (Theorem \ref{theo:2}).
By extending both graphs from Theorem \ref{theo:2} with edge weights $\mrd[i]$, we obtain $\rngi{i}{i} \subseteq \rngi{i}{\maxk}$. Hence, $MST(\mrg[i]) \subseteq \rngi{i}{\maxk}$.
\end{proof}

\subsection{RNG Computation}
\label{subsec:rng-computation}

The performance gain when running HDBSCAN* w.r.t. all values of $\minpoints \in \mptsset$ by using $\rng{\maxk}$ instead of the complete graph $G$ of a dataset relies on a number of factors: the complexity of the additional time to construct $\rngi{i}{\maxk}$ (recall that $G$ does not have to be explicitly constructed), the number of edges in $\rngi{i}{\maxk}$ compared to $G$, and the number of hierarchies $\maxk$ to be computed.

The naive way to compute an RNG for a set of points $\textbf{X}$ is to check for every pair of points $p$, $q$ $\in$ $\textbf{X}$ and each point $c$, whether $c$ is inside $lune(p,q)$. This algorithm runs in $O(n^3)$ time, which is inefficient for large datasets. More efficient strategies are surveyed in \cite{163414}. 

We adopt the approach in  \cite{DBLP:journals/comgeo/AgarwalM92} ---which has sub-quadratic expected time complexity under the assumption that points are in general position--- with an adaptation of the definition of well-separated pairs proposed in \cite{callahan1995decomposition}. This approach has three main steps.

In the first step, the entire dataset is decomposed recursively into smaller and smaller subsets (see \cite{callahan1995decomposition} for details), so that all pairs of obtained subsets $(A,B)$ are \emph{well-separated}. 
The notion of well-separability requires the smallest possible distance between any point $a \in A$ to any point $b \in B$ to be larger than the largest possible distance between points within each of the two sets. For efficiency reasons one does not compute pairwise distances, but instead uses ``safe'' bounds that can be efficiently computed to determine well-separability of two sets. The distance is in our case the mutual reachability distance $\mrd$, and the smallest possible $\mrd$ between two point sets $A$ and $B$ is, because of the max function in the definition of $\mrd$, the shortest possible Euclidean distance between a point $a \in A$ and a point $b \in B$. This distance $D(A,B)$ can be bounded, as in \cite{callahan1995decomposition}, by the distance between the smallest enclosing balls $B_A$ and $B_B$ around the minimum bounding hyper-rectangles enclosing $A$ and $B$, respectively. 
Then, we can define that $A$ and $B$ are well-separated if:
\begin{equation*}
\begin{split}
&D(A, B) \geq \\
&\quad s \cdot \max\{diameter(B_A), diameter(B_B), \max_{p \in A \cup B}(\cd(p))\} 
\end{split}
\end{equation*}
$max\{diameter(B_A), diameter(B_B), \max_{p \in A \cup B}(\cd(p))\}$ represents a bound on the largest possible mutual reachability distance within the sets $A$ and $B$. The separation factor $s > 0$ determines how far both sets have to be from each other to be considered \textit{well-separated}.
The larger the separation factor, the larger the number of generated pairs. For $0 < s <1$, there is no guarantee that the resulting graph will contain the MST edges, and
hence we adopt $s=1$.

	
	
		
		
		
		
		
	
	
	


In the second step, all the well-separated pairs are connected with edges such that a supergraph of the RNG, which we will call \unfilteredRNG, is obtained. 
For each pair $(A, B)$, the points $a_i \in A$ and $b_j \in B$ are connected with an edge if they are Symmetric Bichromatic Closest Neighbors (SBCN), \emph{i.e.}, if there is no other point in $B$ that is closer to $a_i$ than $b_j$ and vice versa. For example, in Figure \ref{fig:bcn}, the points $a_3$ and $b_3$ are SBCN and thus the edge $(a_3, b_3)$ is part of the \unfilteredRNG.

The third step of the RNG computation consists of filtering \unfilteredRNG\ to remove edges that are not in the RNG.
Although \unfilteredRNG\ has typically far fewer edges than the complete graph $G$, a naive filtering approach, which checks for each edge $(a,b)$ in \unfilteredRNG whether each point $c$ is in $lune(a,b)$, can be extremely time consuming for large datasets. 
Therefore, we propose an alternative strategy using information that is computed anyway for HDBSCAN*, which can make the overall filtering process more efficient. It is based on the intuition that points closer to $a$ or $b$ are more likely in $lune(a,b)$ than points that are farther away. For computing multiple HDBSCAN* hierarchies, we initially compute all the needed core distances by performing a single $\maxk$-nearest neighbor query for each point, which means that we do find the $\maxk$ closest points to each point. To support our pruning strategy, we only have to store in addition to the $\cd[i]$ values also the actual $\minpoints$-nearest neighbors. Using this information, for each edge $(a, b)$ with weight $w$, we first check if any of the $\minpoints$-nearest neighbor of $a$ and $b$ is inside $lune(a, b)$. As soon as we find one that is inside, we can safely remove the edge without further checking. If none of those neighbors is inside $lune(a, b)$, we check if $w$ is equal to the core-distance of $a$ or $b$. If that is the case (say for $a$), we know that no other point can be in $lune(a, b)$ (since $lune(a, b)$ is a subset of the ball around $a$ with radius $w$ and we have checked all points inside this ball); hence we know without further checking that the edge is in the RNG. We can choose to perform only these $2 \times \maxk$ checks per edge to obtain a graph, which we call \smartRNG, that is smaller than \unfilteredRNG\ but may contain more edges than the \exactRNG. To obtain the exact \exactRNG, we search the entire dataset whenever we cannot exclude or include an edge based on the $2 \times \maxk$ tests, to determine whether or not there is a point in $lune(a, b)$.

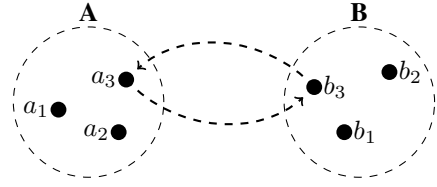
\begin{figure}[t]
	\centering
	
	\begin{tikzpicture}
	
	\draw[dashed] (1.6,1.4) circle [radius=1];
	
	\fill (1.2,1.3) circle[radius=3pt] node[align=left, left] {$a_{1}$};
	\fill (2,1) circle[radius=3pt] node[align=left, left] {$a_{2}$};
	\fill (2.1,1.7) circle[radius=3pt] node[align=left, left] {$a_{3}$};
	
	\draw[dashed] (5.2,1.4) circle [radius=1];
	
	\fill (5,1) circle[radius=3pt] node[align=right, right] {$b_{1}$};
	\fill (5.6,1.8) circle[radius=3pt] node[align=right, right] {$b_{2}$};
	\fill (4.6,1.6) circle[radius=3pt] node[align=right, right] {$b_{3}$};		
	
	\node(A) at (1.6,2.6){\textbf{A}};
	\node(B) at (5.2,2.6){\textbf{B}};
	
	\draw[->, shorten >=0.2cm, bend right=45, thick, dashed] 
	(2.1,1.7) to (4.6,1.6);
	
	\draw[->, shorten >=0.2cm, bend right=45, thick, dashed] 
	(4.6,1.6) to (2.1,1.7);
	\end{tikzpicture}
	\caption{Symmetric Bichromatic Closest Neighbor (SBCN)}
	\label{fig:bcn}
\end{figure}


\algrenewcommand\algorithmicindent{1.3em}%
\begin{algorithm}
\small
\caption{}
\label{alg:rng-hdbscan}
\begin{algorithmic}[1]
	\Require $X$: dataset; 
    		 $n$: $\left\vert{X}\right\vert$;
             $\mptssequence$: $\minpoints$ range;
             $T$: graph to be computed (RNG, RNG, RNG**);
	\For{$i \in \{1,...,n\}$} 							\label{line:core-distances-begin}
        \For{$j \in \mptsset$}
        	\State $M$[$i$][$j$] $\gets$ $(j^{th}$-NN$(i)$, $\cd[j](i));$ \label{core-distances}
        \EndFor
    \EndFor												\label{line:core-distances-end}
    \State 
    \State $wspd \gets WSPD(X, M);$ 					\label{line:wspd}
	\State 
    \For{$(A, B) \in wspd$} 									\label{line:construct-begin}
       	\State $\rng{\maxk} \gets \rng{\maxk} \cup SBCN(A,B);$  		\label{line:construct-end}
    \EndFor
	\State 
	\State $remove \gets \texttt{False};$
    \State 
	\If{$T \neq$ RNG**}												\label{line:test-rng**}
	\For{$(a, b) \in \rng{\maxk}$} 									\label{line:filter-begin}
          \For{$x \in$ $M$[$a$] $\cup$ $M$[$b$]}
              \If{$x \in lune(a,b)$}
                  \State $remove \gets \texttt{True};$
                  \State \textbf{break;}
              \EndIf
          \EndFor
          \If{$\neg$ $remove$}
			\If{$\mrd[\maxk](a,b) = max\{\cd[\maxk](a), \cd[\maxk](b)\}$}
                  \State $remove \gets \texttt{False};$
                  \State \textbf{continue;}            	
            \EndIf
          \EndIf
		\If{$\neg$ $remove$ and $T =$ RNG}
              \For{$x \in X$}									\label{line:sequential-begin}
                  \If{$x \in lune(a,b)$}
                      \State $remove \gets \texttt{True};$
                      \State \textbf{break;}
                  \EndIf
              \EndFor											\label{line:sequential-end}
          \EndIf

          \If{$remove$}
              \State $\rng{\maxk} \gets \rng{\maxk} \setminus (a,b);$
              \State $remove \gets \texttt{False};$
	    \EndIf
    \EndFor														\label{line:filter-end}
	\EndIf
	\State 
	\For{$\minpoints \in \mptsset$}								\label{line:hierarchies-begin}
        \State $MST_{\minpoints} \gets MST(\rngi{\minpoints}{\maxk});$
    	\State $compute\texttt{-}hierarchy(MST_{\minpoints});$
    \EndFor														\label{line:hierarchies-end}
\end{algorithmic}
\end{algorithm}

The pseudo-code for the overall strategy is shown in Algorithm \ref{alg:rng-hdbscan}.
It takes as input a dataset $\textbf{X}$ with $n$ points, a range of $\minpoints$ values, $\mptssequence$, and the type $T$ of the RNG to be computed. All the core-distances for each point $x \in \textbf{X}$ and its corresponding $k$-NN neighbors, for $k \in \mptsset$, are computed in Lines \ref{line:core-distances-begin}-\ref{line:core-distances-end}.
It is important to emphasize that a single $\maxk$-NN query is performed for each $x \in \textbf{X}$.
The statement in Line \ref{core-distances} illustrates the format of the entries of the matrix $M$.
Next, the Well-Separated Pairs Decomposition (WSPD) is performed in Line \ref{line:wspd}. 
In Lines \ref{line:construct-begin}-\ref{line:construct-end}, the \unfilteredRNG\ is constructed by adding one edge for each of the Symmetric Bichromatic Closest Neighbors (SBCN) between the pairs $(A, B) \in wspd$.
The edge filtering occurs between Lines \ref{line:filter-begin}-\ref{line:filter-end}. In case the \unfilteredRNG\ is chosen, the filtering process is completely skipped (Line \ref{line:test-rng**}). Otherwise (both \smartRNG\ and \exactRNG), the filter steps based on the $\maxk$-nearest neighbors are performed. The last filter, based on the sequential scan of the dataset (Lines \ref{line:sequential-begin}-\ref{line:sequential-end}), is only performed when the \exactRNG\ is to be computed.
At the end (Lines \ref{line:hierarchies-begin}-\ref{line:hierarchies-end}), the MSTs and hierarchies are computed for all the values of $\minpoints \leq \maxk$, using the computed RNG.

Note that, if one has to compute just a single MST from $G$, even though this operation is quadratic in the number of points since it depends on the number of edges in $G$, it may not pay off to first construct $\rng{\maxk}$. However, as the number of hierarchies that have to be computed increases, the initial overhead of constructing $\rng{\maxk}$ can substantially speed up the overall time to complete the $\maxk$ MSTs, if the number of edges in $\rng{\maxk}$ is much smaller than in $G$.

%
%
\section{Experiments}
\label{sec:experiments}

We conducted experiments to evaluate the efficiency of the proposed method with respect to changes in size and dimensionality of the dataset, and, most importantly, with respect to the number of hierarchies to be computed. 
We also show the \exactRNG, \smartRNG, and \unfilteredRNG\ sizes in comparison to the size of the $\mrg$, since the difference in the number of edges is the source of our performance gain.

To the best of our knowledge, there is no other strategy in the literature that aims at computing multiple hierarchies efficiently.
Thus, we compare our strategy to a straightforward baseline that runs HDBSCAN* multiple times, one for each $\minpoints$ value in the given range, but with the optimization of pre-computing the core distances for all points, in the same way we do in our approach (see Section \ref{sec:approach}), so that $kNN$ queries are only executed once and not for each value of $\minpoints$.

To study the computational trade-offs of the different edge filtering strategies described in Subsection \ref{subsec:rng-computation}, we show results for three variants: \unfilteredRNG-HDBSCAN*, which just uses the \unfilteredRNG\ without any additional filtering; \smartRNG-HDBSCAN*, which applies only the filtering based on $\maxk$ nearest neighbors; and \exactRNG-HDBSCAN*, which applies the complete filtering to obtain the exact \exactRNG. 

All methods have been implemented on top of the original HDBSCAN* code, provided by the authors of \cite{DBLP:journals/tkdd/CampelloMZS15}, in Java. The core-distances are computed with the aid of a $K$d-Tree index structure, adapted from \cite{kdtreeimplementation}. The experiments were performed in a virtual machine with 64GB RAM, running Ubuntu. For runtime experiments, we measure the total running time to compute core-distances and MSTs, and report the average runtime over 5 experiments.

The datasets were obtained using the generator proposed in \cite{handl2005cluster}, varying the number of dimensions from 2 to 128, the number of points from 16k to 1M, and the value of $\maxk$ from 2 to 128.  
Table \ref{tab:experiments} shows these values and indicates in bold the default value for each variable when other variables are varied.

\begin{table}[!t]
	\renewcommand{\arraystretch}{1.4}
	\caption{Experimental Setup}
	\label{tab:experiments}
	\centering
	\begin{tabular}{c|c}
		\toprule
		\textbf{Variables} & \textbf{Values}                         		   \\ \midrule
		\#points		    & 16k, 32k, 64k, \textbf{128k}, 256k, 512k, 1M	   \\
		\#dimensions      & 2, 4, 8, \textbf{16}, 32, 64, 128                \\
		$\maxk$        & 2, 4, 8, \textbf{16}, 32, 64, 128     		   \\ \bottomrule
	\end{tabular}
\end{table}




\subsection{Effect of Dataset Size}
\begin{figure*}[!htb]
	\centering
	\subfloat[\#points vs. Runtime]{\includegraphics[width=.33\textwidth]{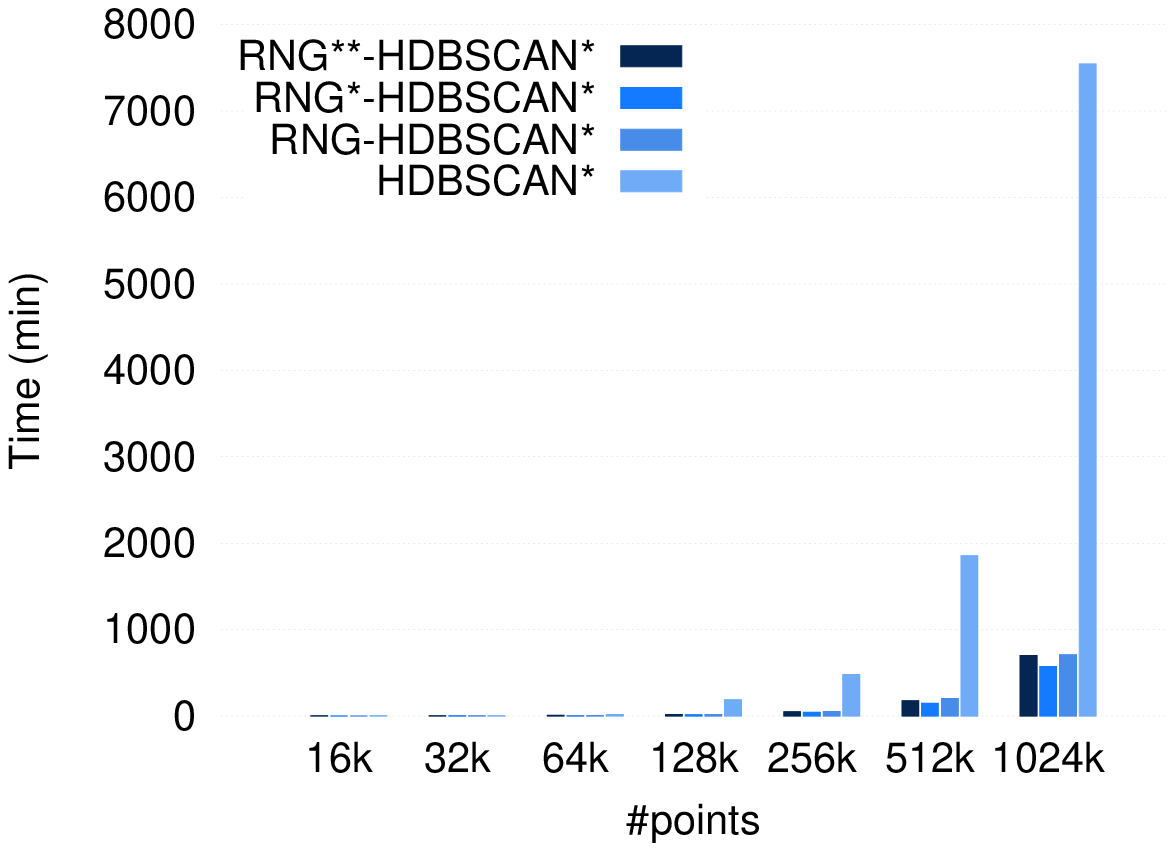} \label{fig:exp-dataset}}
	\subfloat[\#dimensions vs. Runtime]{\includegraphics[width=.33\textwidth]{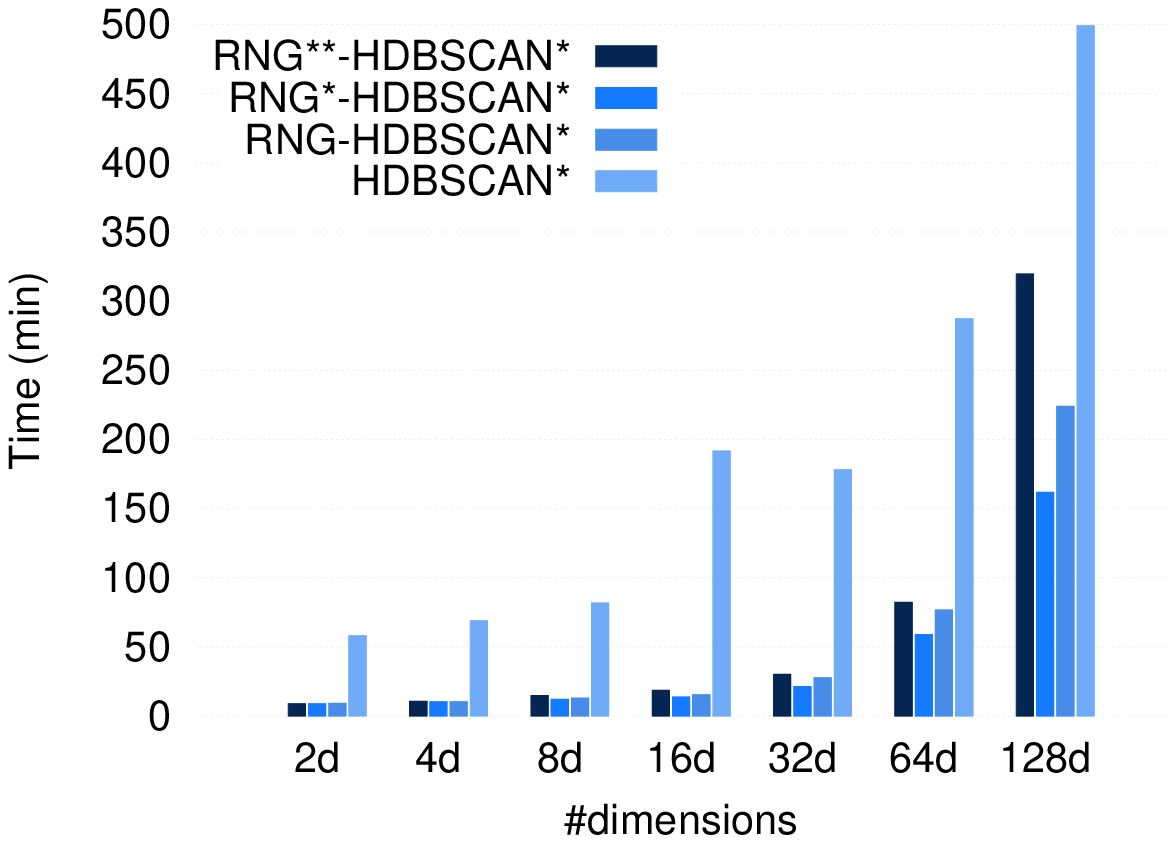} \label{fig:exp-dimensionality}}
	\subfloat[$\maxk$ vs. Runtime]{\includegraphics[width=.33\textwidth]{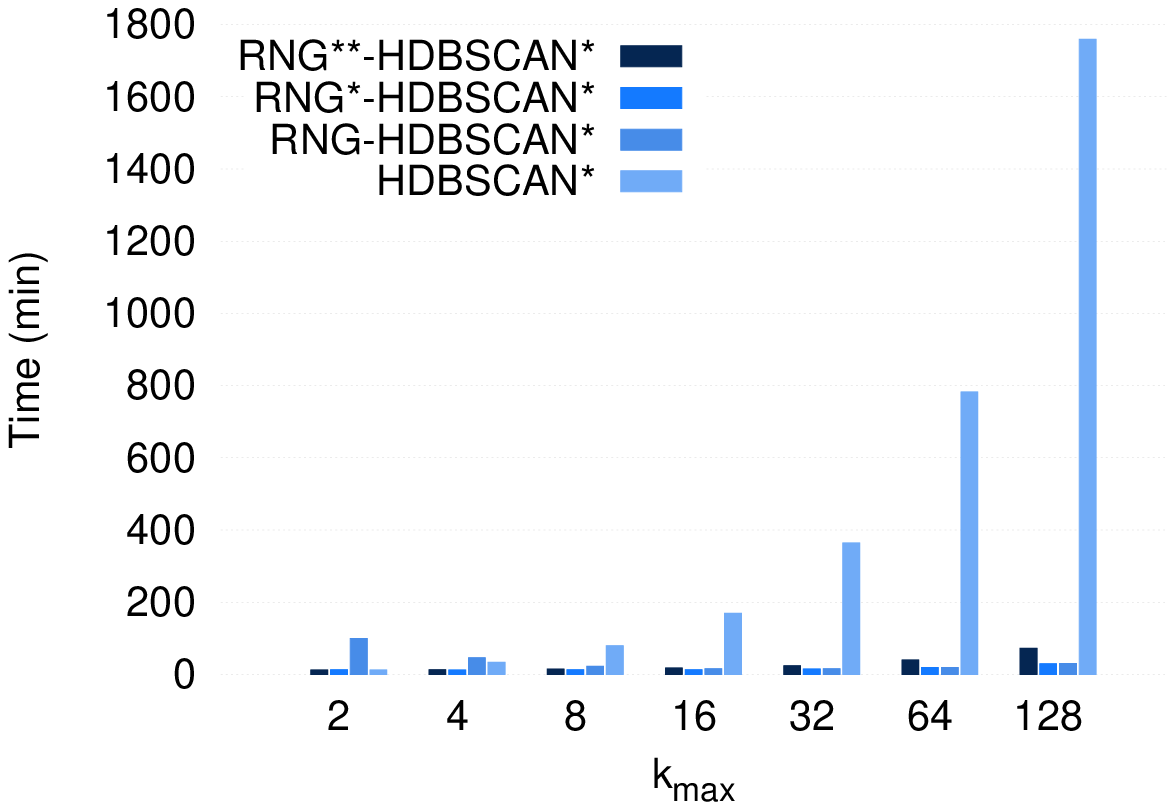} \label{fig:exp-minpoints}}
	\caption{Runtime as a function of the dataset size, dataset dimensionality, and $\maxk$. (Note that the x-axis is in log scale.)}
	\label{fig:performance}
\end{figure*}

\begin{figure*}[!htb]
	\centering
	\subfloat[\#points vs. RNG size]{\includegraphics[width=.33\textwidth]{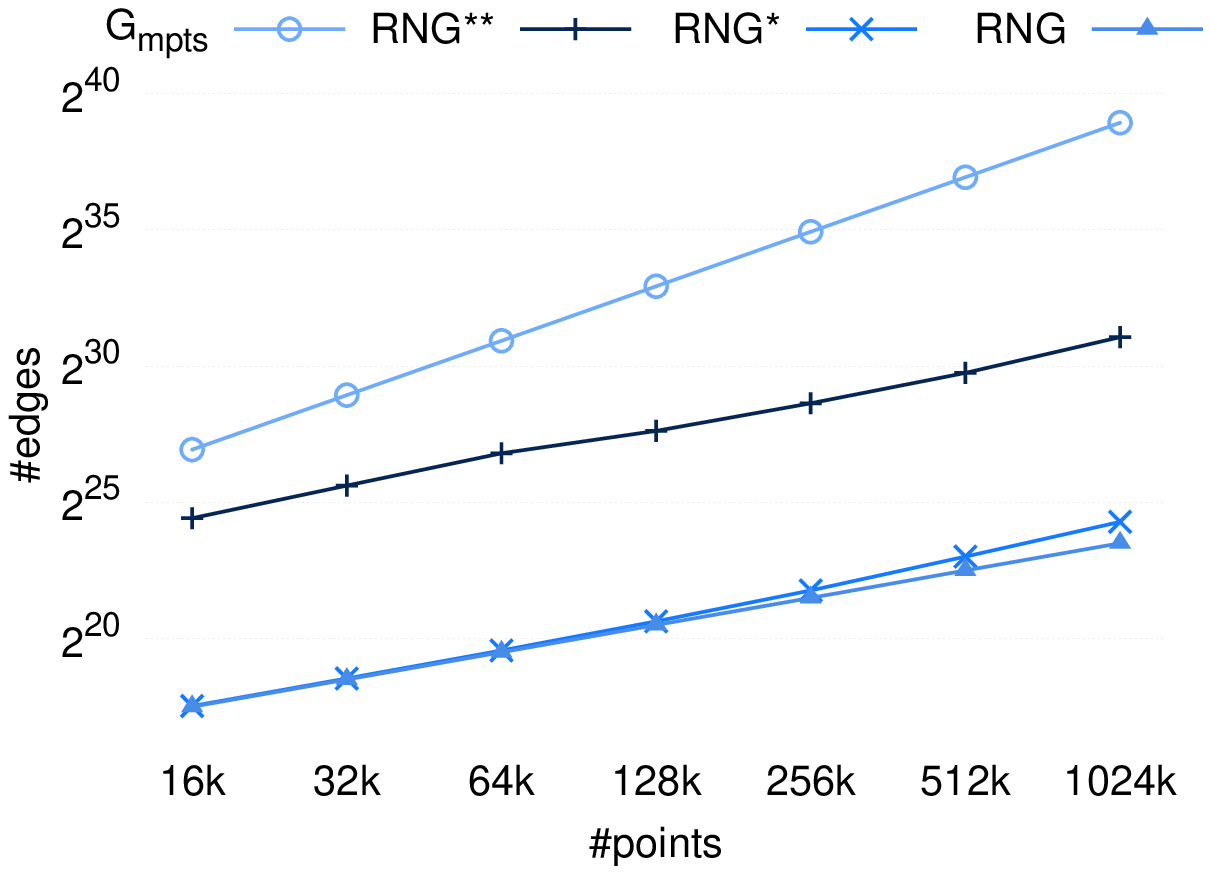} \label{fig:rng-exp-dataset}}
	\subfloat[\#dimensions vs. RNG size]{\includegraphics[width=.33\textwidth]{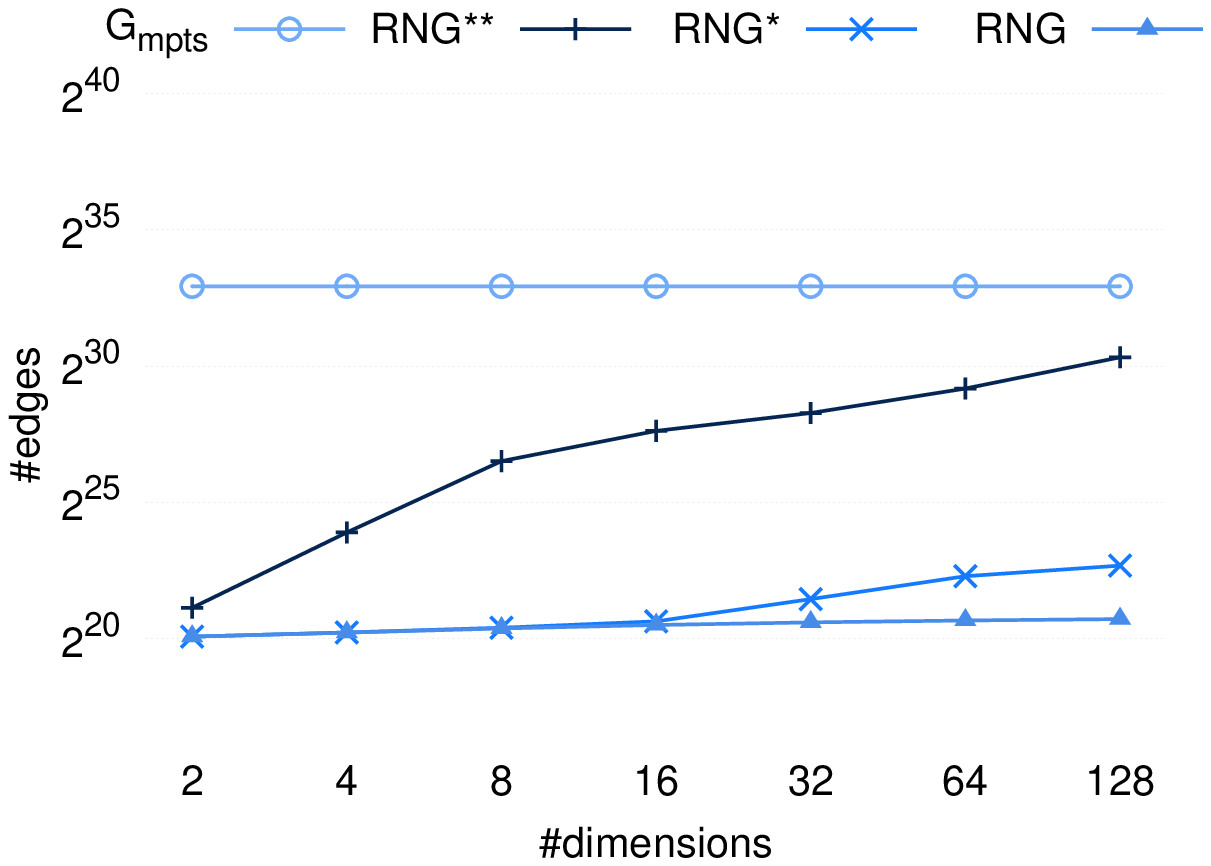} \label{fig:rng-exp-dimensionality}}
	\subfloat[$\maxk$ vs. RNG size]{\includegraphics[width=.33\textwidth]{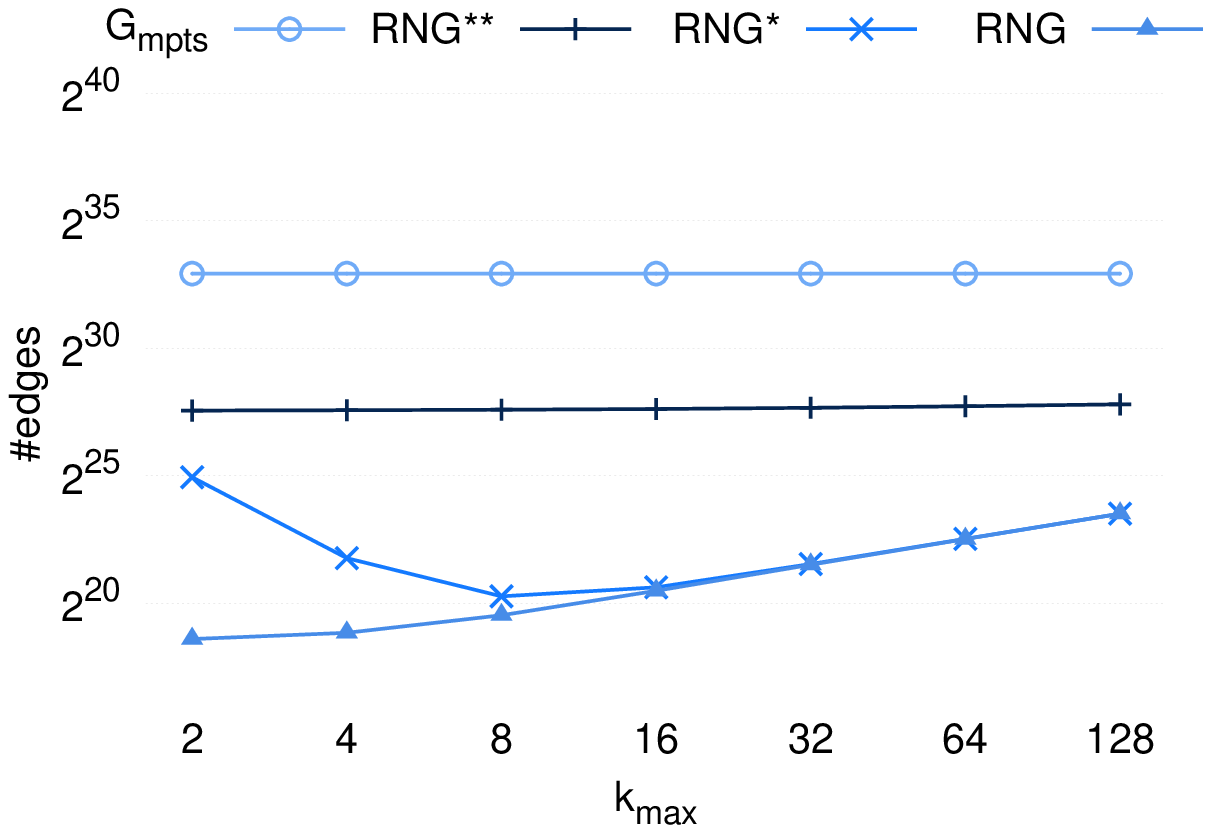} \label{fig:rng-exp-minpoints}}
	\caption{RNG size as a function of the dataset size, dataset dimensionality, and $\maxk$. (Note that both axes are in log scale.)}
	\label{fig:rng}
\end{figure*}

\begin{figure*}[!htb]
	\setlength\abovecaptionskip{-7pt}
	\centering
\includegraphics[width=\textwidth]{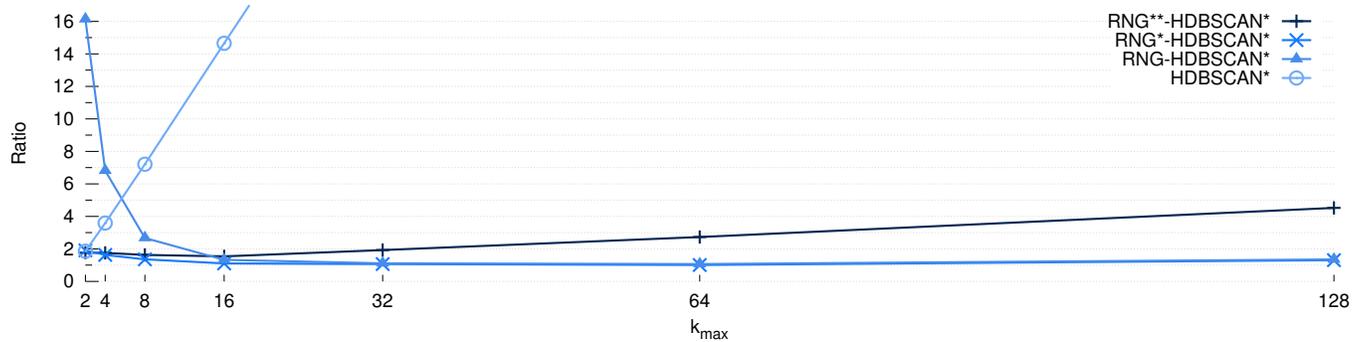} \label{fig:minpoints-ratio}
	\caption{Ratio: runtime to compute $\maxk$ MSTs/hierarchies divided by the runtime to compute a single MST/hierarchy.}
	\label{fig:minpoints-ratio}
\end{figure*}

Figure \ref{fig:exp-dataset} shows the total runtime as a function of the dataset size with default values for the remaining variables (i.e., computing 16 MSTs in 16-dimensional datasets). As expected, the runtime tends to increase as the number of points increases for all methods. For datasets up to 64k points, all strategies have similar performances, but for larger datasets the difference between our approaches and the baseline increases significantly as the number of points increases. For 128k points, the baseline strategy already takes approximately twice as much time as our approaches. For 1024k points, we actually interrupted each run of the baseline before it finished. 

Figure \ref{fig:rng-exp-dataset} shows the number of edges in $\mrg$, \unfilteredRNG, \smartRNG, and RNG, as a function of the dataset size. 
As expected, the number of edges increases with the number of points. However, the RNGs are significantly smaller than the complete graph for all dataset sizes. In fact, even for the largest dataset, the sizes of the \smartRNG\ and \exactRNG\ are smaller than the size of the $\mrg$ for the smallest dataset.

When comparing \exactRNG\ and \unfilteredRNG, the time spent filtering to obtain the exact \exactRNG\ is compensated by a smaller graph size, which in turn results in faster MST computations. This explains why both \exactRNG\ and \unfilteredRNG\ exhibit similar running times in Figure \ref{fig:exp-dataset}, despite the differences in their sizes. Only when the partial fast filter based on $k$-nearest neighbors is applied to obtain \smartRNG, the total runtime is faster. This is because the filter is effective in producing a graph that is almost as small as the exact \exactRNG, yet in less time. 

\subsection{Effect of Dimensionality}


Figure \ref{fig:exp-dimensionality} shows the effect of dataset dimensionality on the runtimes. As expected, all approaches are affected by increasing dimensionality, since most of the underlying techniques for clustering, $k$NN queries, and RNG computation are bound to eventually become less effective as dimensionality increases. This is due to a number of effects that are generally referred to as ``curse of dimensionality.'' However, since our datasets do contain cluster structures, these effects are not critically severe even in 128 dimensions.   

We can observe that all RNG-based strategies perform better than the baseline in all datasets, but as dimensionality increases, the difference between the unfiltered RNG (\unfilteredRNG) and the filtered versions (\smartRNG\ and \exactRNG) increases. This can be explained by looking at the number of graph edges, as shown in Figure \ref{fig:rng-exp-dimensionality}. 
The size of the exact \exactRNG\ is barely affected by an increase of dimensionality in these datasets, while the unfiltered \unfilteredRNG\ exhibits a pronounced growth in the number of edges, which approaches the complete graph $\mrg$. This shows that the generation of well-separated pairs is very sensitive to dimensionality, becoming less effective in implicitly excluding edges that cannot be in an RNG. On the other hand, the exact relative neighborhood graph still has significantly fewer edges than a complete graph in these scenarios  ---although, theoretically, it also must eventually approach the complete graph \cite{163414,callahan1995dealing}. 

Notably, the number of edges for \smartRNG\ increases only slightly as the dimensionality increases, which shows that the pruning strategy using only the pre-computed $k$-nearest neighbors (16, as $\maxk = 16$ in this experiment) stays quite effective, even in the 128-dimensional datasets, resulting in the best runtime performance overall.


\subsection{Effect of Upper Limit $\maxk$}

Figure \ref{fig:exp-minpoints} and Table \ref{tab:maxk} show the runtimes w.r.t. $\maxk$. The runtime of all our methods is very low compared to the baseline, for which runtime increases linearly, as expected.
%
The runtime of HDBSCAN*-\unfilteredRNG\ increases very slightly with $\maxk$ as also the number of edges increases slightly, but stays significantly below the number of edges in $\mrg$, as shown in Figure \ref{fig:rng-exp-minpoints}. 

\exactRNG-HDBSCAN* shows a slightly higher runtime for $\minpoints=2$, which then decreases for $\minpoints=4$ and $\minpoints=8$, after which it stays almost constant and becomes almost indistinguishable in performance to \smartRNG-HDBSCAN*. 
\smartRNG-HDBSCAN*, which only uses the $\maxk$-nearest neighbors of objects for pruning \unfilteredRNG, shows the most stable runtime behavior; its increase in runtime as $\maxk$ increases is almost unnoticeable. For the largest value of $\maxk$, its difference in runtime to the baseline method corresponds to a speedup of about 60 times.
The runtime behavior of \exactRNG\ and \smartRNG\ can be explained by the number of edges in \exactRNG\ and \smartRNG, shown in Figure \ref{fig:rng-exp-minpoints}. For $\minpoints=2$, the number of edges in \smartRNG\ is much larger than in \exactRNG\ (while still being smaller than in \unfilteredRNG). The reason is that the  the filtering strategy based on the $\maxk$-NN is not yet very effective when only the two nearest neighbors are considered. In this case, for all the edges that are removed from \smartRNG\ to obtain \exactRNG, a sequential scan has to be performed, which is overall more costly in terms of runtime than the gain in runtime for computing the MST of \exactRNG\ with fewer edges. These results also show that (1) computing MSTs is very fast, compared to the rest of the computation, if the underlying graphs are already relatively small compared to the complete graph, and (2) our pruning heuristic based on $\maxk$-NNs becomes more effective as $\maxk$ increases, leading to an almost indistinguishable performance between \exactRNG\ and \smartRNG\ for $\maxk \geq 16$.

While the observed speedups are  impressive, the significance of our contribution becomes even more clear, if we look at the runtime from a different perspective. Figure \ref{fig:minpoints-ratio} shows the ratio of the runtime to compute $\maxk$ MSTs over the runtime to compute a single MST.
\smartRNG\ exhibits a very stable ratio of about 2 for all values of $\maxk$, {\em i.e}, we can use it to compute as many as 128 MSTs/hierarchies for the computational cost of naively computing about 2 MSTs/hierarchies. 

\begin{table}[]
\centering
\caption{$\maxk$ vs. Runtime (min.)}
\label{tab:maxk}
\tabcolsep=0.11cm
\resizebox{\columnwidth}{!}{%
\begin{tabular}{ccccc}
\toprule
\multicolumn{1}{l}{$\maxk$} & HDBSCAN*    & RNG**-HDBSCAN* & RNG*-HDBSCAN* & RNG-HDBSCAN* \\ \midrule
2                    & \textbf{12} & \textbf{12}    & \textbf{12}   & 99           	\\
4                    & 33          & \textbf{12}    & \textbf{12}   & 45           	\\
8                    & 79          & 14             & \textbf{12}   & 22           	\\
16                   & 169         & 17             & \textbf{13}   & 15           	\\
32                   & 363         & 23             & \textbf{14}   & 15  		   	\\
64                   & 781         & 40             & \textbf{18}   & 19           	\\
128                  & 1759        & 72             & \textbf{29}   & 30			\\
\end{tabular}}
\end{table}

\section{Conclusion}
\label{sec:conclusion}
We presented an efficient strategy for computing multiple density-based clustering hierarchies. We formally showed that the use of the Relative Neighborhood Graph as a substitute for the Mutual Reachability Graph is advantageous when one wants to explore state-of-the-art HDBSCAN* solutions w.r.t. multiple values of $\minpoints$, while ensuring theoretical correctness of results. 
Our experiments showed that the proposed method can be significantly
faster than a baseline strategy based on running HDBSCAN* exhaustively (yet in an optimized way) for different values of $\minpoints$. In particular, it scales considerably better when running on large datasets and more prominently for broader ranges of $\minpoints$ values. 

In our future work we intend to investigate strategies to simultaneously explore, visualize and possibly combine the whole spectrum of clustering solutions that are available both across multiple hierarchies (corresponding to different values of $\minpoints$) as well as across different hierarchical/density levels, taking into account the quality of these solutions according to different unsupervised criteria.

%
%
\section*{Acknowledgment}
\label{sec:acknowledgments}
Research partially supported by NSERC, Canada, and by CNPq, under the Program Science without Borders, Brazil.
%
%
\bibliographystyle{abbrv}
\bibliography{ihdbscan}

\end{document}